\documentclass{article}

\usepackage{amsmath, amsfonts, amssymb, amsthm}

\newtheorem{theorem}{Theorem}[section]
\newtheorem{lemma}[theorem]{Lemma}

\theoremstyle{definition}

\newtheorem{example}[theorem]{Example}
\newtheorem{remark}[theorem]{Remark}

\newcommand{\score}{\mathrm{score}}
\newcommand{\ttt}{t}
\newcommand{\ttau}{\tau}

\title{List decoding of repeated codes\thanks{This research was partially supported by the National Science Foundation under Grant No. CCF-0916492, by the Danish National Research Foundation and the National Science Foundation of China (Grant No.11061130539) for the Danish-Chinese Center for Applications of Algebraic Geometry in Coding Theory and Cryptography, by the Spanish grant No. MTM2007-64704, and by the Spanish MINECO under grant No. MTM2012-36917-C03-03.}}

\author{F. Hernando\thanks{Department of Mathematics, Universidad Jaume I, Spain. {\ttfamily carrillf@mat.uji.es}} \and M. O'Sullivan \thanks{Department of Mathematics and Statistics, San Diego State University, USA. {\ttfamily mosulliv@math.sdsu.edu}} \and D. Ruano \thanks{Department of Mathematical Sciences, Aalborg University, Denmark. {\ttfamily diego@math.aau.dk}}}

\date{}

\begin{document}
\maketitle

\begin{abstract}
Assuming that we have a soft-decision list decoding algorithm of a linear code, a new hard-decision list decoding algorithm  of its repeated code is proposed in this article. Although repeated codes are not used for encoding data, due to their parameters, we show that they have a good performance with this algorithm.  We compare, by computer simulations, our algorithm  for the repeated code of a Reed-Solomon code against a decoding algorithm of a Reed-Solomon code. Finally, we estimate the decoding capability of the algorithm for Reed-Solomon codes and show that performance is somewhat better than our estimates.
\end{abstract}
 
\section{Introduction}

List decoding  was introduced by Elias \cite{eli} and Wozencraft \cite{woz}. A list decoder can produce several candidate codewords near the received vector, thus relaxing the requirements of unique decoding and enabling the possibility of decoding beyond half of the minimum distance. The efficient list decoding problem was unsolved for many years until Sudan \cite{Sudan} provided an algorithm for low rate Reed-Solomon codes. Later Guruswami and Sudan \cite{Guruswami-Sudan} gave a general answer for Reed-Solomon codes that has 2 steps: the interpolation step and the root finding step. 

A soft-decoding algorithm works under the assumption that the output of the channel is probabilistic information for the reliability of input data in contrast to a hard-decision list-decoding algorithm which assumes that a word of the ambient space is received.  Koetter and Vardy created a soft-decision list decoding algorithm for RS codes based on the interpolation techniques of Guruswami-Sudan \cite{koetter}. This algorithm has 3 steps.  First, probabilistic information from the channel is translated into an assignment of multiplicities to points in the plane---the points representing received location-value pairs. The other two steps are the  interpolation and root finding steps. 

In this paper we consider the repeated code of a linear code $C \subset \mathbb{F}^n_q$,  that is  $C^\ell = \{ (c, \ldots , c) : c \in C \}\subset  \mathbb{F}_q^{\ell n}$, for some $\ell \ge 2$.  We present a (hard-decision) list decoding algorithm for repeated codes, based on a soft-decision list decoding algorithm for the constituent code, $C$. In our algorithm, the multiplicity step is based on the algebraic structure of the code, rather than the information from the channel. Namely, it is an interpolation problem for a single block taking into account the information provided by the whole codeword. As far as the authors know, this is a novel idea in interpolation decoding. We consider two multiplicity assignment methods in detail, each of them gives rise to a different decoding algorithm. The first choice maximizes the score and, therefore, the error correction capability. The second one minimizes the sum of the multiplicities, thus it minimizes the computational time complexity. 
 
Repeated codes do not have good parameters, but we remark that the minimum distance of a linear code is only an estimate of  unique-decoding capability, and it is an even rougher estimate of the list-decoding performance of the code.  We show that repeated codes of Reed-Solomon codes using the algorithm in \cite{Lee-Michael1} and taking into account our new set up for the multiplicities, may have similar decoding capability to a Reed-Solomon code, with the same information ratio and length, using the algorithm in \cite{Lee-Michael}. However, significant differences in the computation time are observed since the complexity of the first simulation depends on $n$ instead of  $\ell n$, the length of the code. We shall compare our hard-decision list-decoding algorithm to another hard-decision list-decoding algorithm, since it does not make sense to compare it to a soft-decoding algorithm.

Finally, we estimate the decoding capability of the algorithm for Reed-Solomon codes. Even though the bounds assume certain properties of the error vector, they are relatively close to experimental values in the examples.

The paper is organized as follows: In Section~\ref{sec:interpolating} we recall the soft-decision interpolation problem for RS codes.  In Section~\ref{sec:interp-repe} we introduce our list decoding algorithm for repeated codes.  In Section~\ref{sec:sim}, we present some simulations of the decoding algorithm with MAGMA \cite{ma} and interpret the results. We estimate the number of errors $\ttt$ that we can decode in Sections~\ref{sec:boundsfortau2} and~\ref{sec:boundsfortau}. Section~\ref{sec:con} concludes the article.

\section{Soft-Decoding}\label{sec:interpolating}

Koetter and Vardy discovered a soft-decision list decoding algorithm for RS codes based on the interpolation techniques of Guruswami-Sudan \cite{koetter},  and   later the interpolation step was described using Gr{\"o}bner bases in \cite{Alekhnovich,Lee-Michael}. For our simulations, we use  the algorithms in \cite{Lee-Michael1,Lee-Michael} , which use the same approach, the only difference being  the multiplicity assignment. We recall in this section the soft-decoding algorithm in  \cite{Lee-Michael1}  for Reed-Solomon codes.

Let $\alpha_1, \ldots, \alpha_n$ be $n$ different points of the finite field  $\mathbb{F}_q$ with $q$ elements and let $C$ be the Reed-Solomon code with parameters $[n,k,d]$ defined as $$C=\{(h(\alpha_1),h(\alpha_2),\ldots,h(\alpha_{n})): \deg (h)\leq k-1\}.$$ 
For soft decision decoding, Koetter and Vardy \cite{koetter} use reliability information provided by the channel to assign
multiplicities  $m_{i,\beta}$ to each point $p_{i,\beta}=(\alpha_{i},\beta)$ for $i=1,\ldots,n $ and
$\beta\in\mathbb{F}_q$.
Let $M$ be the collection of these multiplicities, 
\[M=\{(p_{i,\beta},m_{i,\beta}):i=1,\ldots,n;  \beta\in\mathbb{F}_q \}\]
Consider  the ideal in $\mathbb{F}_q[x,y]$ of polynomials interpolating at the points of $M$ with the desired multiplicities:
\begin{equation}\label{def: I_M}
I_M=\{f\in \mathbb{F}_q[x,y]: \mathrm{mult}_p (f)\geq m \ for \ (p,m)\in M\},
\end{equation} {where $\mathrm{mult}_p (f)$ denotes the multiplicity of $f$ at $p$.}
For  $r=(r_1,\ldots,r_n)\in\mathbb{F}_q^n$, let $h_r$ be the interpolating polynomial at the points
$(\alpha_{i},r_i)$ for $i=1,\ldots,n$.
The key observation of Guruswami and Sudan is that for $f  \in
I_M$ and a codeword $c$, $y-h_c$ is a factor of $f(x,y)$ when 
\begin{align}
\label{e:keycond}
\sum_{i=1}^n m_{i, c_i} > \deg (f(x, h_c))
\end{align}

For a given  $M$ we therefore define for each $r=(r_1,\ldots,r_n)\in\mathbb{F}_q^n$ 
$$
\score(r)=\sum_{i=1}^n m_{i,r_i}.
$$
We use the $(1,k-1)$-weighted degree of polynomial $f = \sum f_{i,j}x^iy^j$, which is defined
to be $\deg_{1,k-1} (f) = \max\{i+j(k-1): f_{i,j} \ne 0\}$.  
Extend the weighted degree to a monomial ordering by taking
$x^ay^b>_{k-1}x^iy^j$ if either
$\deg_{1,k-1}(x^ay^b)>\deg_{1,k-1}(x^iy^j)$ or if
$\deg_{k-1}(x^ay^b)=\deg_{k-1}(x^iy^j)$ and $b>j$.
We may now interpret \eqref{e:keycond} as saying that  a polynomial
$f(x,y)$ such that the score of $c$ is larger than the
$(1,k-1)$-degree of $f(x,y)$ is divisible by  $y-h_c$.
Factoring $f$ would  produce the high scoring codeword $c$.

The requirement that $f(x,y)$ pass through $(\alpha_i,\beta)$ with multiplicity $m_{i,\beta}$ imposes  $\binom{m_{i,\beta}+1}{2}$ conditions, so overall we have 
\begin{equation}\label{eq:nequat}
N= \sum_{i=1}^{n}\sum_{\beta\in\mathbb{F}_q}\binom{m_{i,\beta}+1}{2},
\end{equation}
conditions.  
There is an upper bound for the $(1,k-1)$-weighted degree $d$ of a polynomial $Q(x,y)$ given that $N$ conditions are imposed 
(see e.g. \cite[Proposition 3]{Lee-Michael1}).

Hence, we look for an interpolating polynomial of the
form $\sum_{(i,j)\in S}f_{i,j}x^iy^j$, where $S\subset\{(i,j):i,j\geq
0\}$ is the subset of indices with $(1,k-1)$-degree less than $d$,
and by choice of $d$, $|S| > N$.
The minimal polynomial with respect to the $>_{k-1}$ is usually called $Q(x,y)$.
One  method for obtaining $Q(x,y)$, for example  \cite{Lee-Michael1}, is to  compute a Gr\"obner basis of $I_M$ with respect to the $(1,k-1)$-weighted degree and pick the smallest element in it.

\section{List decoding of repeated codes}\label{sec:interp-repe}

Let $C \subset \mathbb{F}_q^n$ be a Reed-Solomon code with parameters $[n,k,d]$ and generator matrix $G$. We will describe the algorithm for a Reed-Solomon code, however, the algorithm can be extended in a straightforward manner to any linear code provided with soft-decision list-decoding algorithm.

We consider the repeated code of $C$, $$C^\ell = \{ (c, \ldots , c ) : c \in C \},$$which has parameters $[\ell n,k,\ell d]$ and generator matrix $(G|\cdots|G)$ \cite[Problem 17 of Ch. 1]{mac}. We will describe a hard-decision list decoding algorithm for $C^\ell$ by using a soft-decoding algorithm for $C$, thus we only have to define the matrix of multiplicities from a received word.

Let $c=(c_1,\ldots,c_n) \in C$ then a typical codeword of $C^\ell$ is of the form  
$$
\mathbf{c}=(c_1,\ldots,c_n,\ldots,c_1,\ldots,c_n) \in \mathbb{F}_q^{\ell n}.
$$
One can also understand a vector of length $\ell n $ as an $\ell \times n $ matrix. Hence, if $\mathbf{v}\in \mathcal{M}(\ell\times n ,\mathbb{F}_q)$ we  denote by  $v_i^j$ the entry corresponding to the  $j$-th row (block) and the  $i$-th column (position),  for $i=1,\dots,n$ and $j=1,\dots , \ell$. According to this notation a word in $\mathbb{F}_q^{\ell n}$ may be represented as $$\mathbf{v}=(v_1^1,\ldots,v_n^1,\ldots,v_1^{\ell},\ldots,v_n^{\ell}).$$

Let $\mathbf{c}$ be the sent word and $\mathbf{r}=\mathbf{c}+\mathbf{e}$ the received word 
with error weight $\ttt=wt(\mathbf{e})$. We have that $\mathbf{r}=(r_1^1,\ldots,r_n^1,\ldots,r_1^{\ell},\ldots,r_n^{\ell})$ and  $\mathbf{c}=(c_1^1,\ldots,c_n^1,\ldots,c_1^{\ell},\ldots,c_n^{\ell})$. Since $\mathbf{c} \in C^\ell$, one has that
$$
c_i^j=c_i^k, \mathrm{~for~every~}  i\in\{1,\ldots,n\}   \mathrm{~and~}   j,k\in\{1,\ldots,\ell\}. 
$$ 
Therefore, $r_i^j=r_i^k$ if and only if  $e_i^j=e_i^k$. Hence if the received word has the same value in several positions it is likely that the error value in these positions is zero. Namely, the more positions where the values agree the more likely that these positions are error free. Moreover, the bigger the base field $\mathbb{F}_q$ the more likely the previous assumption is right. Based on this fact we will define the multiplicities $m_{i,\beta}$  at the point $p_{i,\beta}=(\alpha_i,\beta)$ for the soft-decoding algorithm of $C$.

For a received word $\mathbf{r}$, we define two different assignment of multiplicities. For $i=1,\ldots,n$ and  $\beta\in\mathbb{F}_q$: 

\begin{itemize}

\item[(1)] $m_{i,\beta}=|\{j\in\{1,\ldots,\ell\} : r^j_i =\beta \}|$, or

\item[(2)] $m_{i,\beta}=1$ if $|\{j\in\{1,\ldots,\ell\} : r^j_i =\beta \}| \ge b$, where $b \in \{1, \ldots, \ell \}$. We will consider $b = \lfloor \ell /2 \rfloor +1$ and $b =\lfloor \ell /2 \rfloor $ for our simulations in Section~\ref{sec:sim}.
\end{itemize}

\begin{example}
Let $\mathbf{r}=(0,0,0, 0,0,0, 0,0,1, 0,1,1, 0,2,2) \in \mathbb{F}_3^{\ell n}$ a received word with $\ell = 5$ and $n=3$, i.e. using matrix notation$$ \mathbf{r}= \left( \begin{array}{ccc}
0 & 0 & 0 \\ 
0 & 0 & 0 \\ 
0 & 0 & 1 \\ 
0 & 1 & 1 \\ 
0 & 2 & 2
\end{array} \right ).$$ 

We consider the multiplicities assignment for this word, the non-zero multiplicities are:

$m_{1,0} = 5, m_{2,0} = 3, m_{2,1}=1, m_{2,2}=1, m_{3,0}=2, m_{3,1}=2, m_{3,2}=1$ for multiplicity assignment (1).

$m_{1,0}=1, m_{2,0}=1$ for multiplicity assignment (2) with $b=3$.

$m_{1,0}=1, m_{2,0}=1, m_{3,0}=1, m_{3,1}=1$ for multiplicity assignment (2) with $b=2$.

\end{example}

For decoding, we will consider the soft-decision algorithm for $C$ (for instance \cite{Lee-Michael1}) with these multiplicities.  Trivially, we have that the computational complexity of this algorithm for $C^\ell$ equals the computational complexity of the soft-decoding algorithm for $C$ with multiplicity assignments corresponding to the choice (1) or (2).

For instance, for the first multiplicity assignment it is $O(R^{\frac{1}{2}}n^2 m^5 )$ using \cite{Lee-Michael1}, where $R=k/n$ is the rate of a block and $m=max\{m_{j,\beta} :  j=1,\ldots,n\mathrm{~and~}\beta\in\mathbb{F}_q\}$. Thus, for the second multiplicity assignment, the complexity is $O(R^{\frac{1}{2}}n^2)$ using \cite{Lee-Michael1}. 
Notice that for decoding a repetition code with parameters $[n\ell,k,\ell d]$, we perform a list decoding of a Reed-Somon code of length $n$, whilst a list decoding of a $[n\ell,k]$ Reed-Solomon code has complexity $O((k/\ell n)^{\frac{1}{2}}(\ell n)^2 m^5 )$ using \cite{Lee-Michael}. That is, our method is about $\ell^{3/2}$ times faster than decoding a Reed-Solomon code with the same parameters.  One may consider  other heuristics for  multiplicity assignment between  (1) and (2), that is,  $0\leq m_{i,\beta}\leq |\{j\in\{1,\ldots,\ell\} : r^j_i =\beta \}|$. We have considered these two  multiplicity assignments since they represent extremes between  (1) maximizing  the score and  (2)  minimizing the time complexity.

\subsection{Multiplicity assignment (1)}

For a word in $C^\ell$, we consider its score using the first multiplicity assignment. 

\begin{lemma}\label{le:scorebound}
Let $\mathbf{c} = (c , \ldots , c)$ be a sent codeword and let $\mathbf{r} = \mathbf{c} + \mathbf{e}$ be a received word  with $wt(\mathbf{e}) = \ttt$. Then $$\score (\mathbf{c}) = \sum_{i=1}^n m_{i,c_i}  = \ell n - \ttt.$$
\end{lemma}
\begin{proof} Note that for $i=1, \ldots, n$,  one has that $m_{i,c_i} = \ell - wt (e^1_i, \ldots , e^l_i)$. Summing over the $n$ positions gives the result. 
\end{proof}

As the  lemma shows,  the score of the received word is a simple
function of the number of errors $\ttt$.
In Section~\ref{sec:boundsfortau} we investigate the interplay between
the number of conditions imposed, equation \eqref{eq:nequat},  and the score,
we derive bounds for successful decoding of  $\ttt$ errors, and we 
compare the bounds with the simulation results presented in Section~\ref{sec:sim}.

\subsection{Multiplicity assignment (2)}

Using the second multiplicity assignment, the score of a received word might be lower than the one obtained using the first multiplicity assignment. Therefore, the error correction capability is smaller. However, the computational complexity is lower since $\max\{m_{j,\beta} :  j=1,\ldots,n\mathrm{~and~}\beta\in\mathbb{F}_q\}$ is equal to $1$. Note that if $m_{i,\beta}=0$ for every $\beta \in \mathbb{F}_q$ then we are considering an erasure at position $i \in \{ 1 \ldots, n \}$. 

We consider $b=\lfloor \ell /2 \rfloor +1$ because if $b \ge  \lfloor \ell /2 \rfloor +1$, then,  for $i \in \{1, \ldots , n \}$,  $m_{i,\beta}=0$ for all $\beta \in \mathbb{F}_q$  but for, at most, one. Hence, we will have  either an erasure (we interpolate with multiplicity zero at that position)  or we interpolate with multiplicity one. With this choice for $b$ the algorithm is very fast.

\begin{remark}
Let us compare our algorithm with some approaches in the bibliography. A decoding algorithm for a repeated code $C^\ell$ can be obtained decoding every block $r_j$ of the received word $r$ for $j=1,\ldots,\ell$ until the decoded block $c'_j$ verifies that $(c'_j , \ldots , c'_j)$ is at distance $\lfloor (\ell d-1)/2 \rfloor$ of the received word. This is the approach of \cite{Lally} for quasi-cyclic codes, that is when $C$ is cyclic. One could also consider $C^\ell$ as a convolutional code \cite{Dumer}. However, all these approaches will have bad performance due to the poor parameters of the repeated codes.
\end{remark}

\section{Computer experiments}\label{sec:sim}

We have compared the performance of our algorithm for the repetition
code of a $[n,k]$ Reed-Solomon code over $\mathbb{F}_{p^v}$  to the
list-decoding algorithm  \cite{Lee-Michael} for a $[n\ell,k]$
Reed-Solomon code over $\mathbb{F}_{p^u}$ where $u$ is chosen so that
$p^u> n\ell$.  The information rate of the two codes is the same,
although the $[n\ell,k]$ Reed-Solomon code uses larger symbol size.  
For each code we tested a range of values for $\ttt$ to determine the
point where correction performance declined.
We also compared the time required for decoding.  In the repeated code
case, we tested  the algorithm described in the previous section for both multiplicity assignments.
For the Reed-Solomon code over the
larger field we used multiplicity~1 because for higher multiplicity,
the algorithm is very slow. 
 
 We have implemented in MAGMA the decoding algorithm described in Section~\ref{sec:interp-repe}: given the codeword $c=(0,\ldots,0)$ in a $[n,k,d]$ Reed-Solomon code we consider the repeated codeword $(c,\ldots,c) \in C^\ell$. In particular we choose $n=63$, $\ell=5$ and several values for $k$. We consider $\ttt$ errors in $\ttt$ uniformly distributed random positions among the $\ell n$ positions, so it is possible that more than one error may occur in position $i$. We apply the algorithm in \cite{Lee-Michael1} with the prescribed multiplicities described using the multiplicity assignment (1) and (2) with $b=\lfloor \ell / 2 \rfloor + 1 = 3$ and $b= \lfloor \ell / 2 \rfloor =2$. We  ran the procedure $10000$  times each for different  values of $\ttt$. We declare success if the sent codeword is in the output list. For multiplicity assignment (2) the list has always size one since we are doing linear interpolation, and for multiplicity assignment (1) all the experiments produce size one as well.

\begin{enumerate}
\item Consider as block code the RS code with parameters $[63,14,50]$ over $\mathbb{F}_{2^6}$.  The repeated code with $\ell=5$ has parameters $[315,14,250]$.  Our simulations in Tables~\ref{RS-63-14} and~\ref{RS-63-14-m1} show that we can uniquely decode about $226$ errors using multiplicity assignment (1), $183$ errors using multiplicity assignment (2) with $b=3$ and $218$ errors using multiplicity assignment (2) with $b=2$. However, the algorithm using  multiplicity assignment (2) with $b=3$ is about 12  times (resp 6 times with $b=2$) faster than the algorithm using multiplicity assignment (1).

We compare the previous algorithms with  list decoding  of the  RS code over $\mathbb{F}_{2^9}$ with parameters $[315,14,302]$.   Using multiplicity one we can decode about the same number of errors as with multiplicity assignment (1) but it is 176 times slower than the one with multiplicity assignment (1) and 1077 times slower than the simulation with multiplicity assignment (2) and $b=2$,  see Table~\ref{RS-315-14}.

\begin{table}[htb]\caption{List decoding $[315,14,250]$ repeated code over $\mathbb{F}_{2^6}$  with constituent  $[63,14,50]$ RS code, multiplicity assignment (1)}
\begin{center}
\begin{tabular}{|c|c|c|c|c|c|c|c|}\hline 
$\ttt$ & $224$ & $225$ & $226$& $227$ & $228$ & $229$& $230$\\ \hline
$\frac{\text{Number of success}}{10000}$   & $1$ & $1$& $1$ & $.9999$& $.9996$ & $.9992$ & $.9989$\\ \hline
Time  & $1435.190$  & $1424.100$ & $1428.570$ & $1324.950$ &  $1329.510$ & $1322.360$ & $1322.220$\\ \hline
\end{tabular}
\label{RS-63-14}
\end{center}
\end{table}

\begin{table}[htb]\caption{List decoding $[315,14,250]$ repeated code over $\mathbb{F}_{2^6}$  with constituent $[63,14,50]$ RS code, multiplicity assignment (2)}
\begin{center}
\begin{tabular}{|c|c|c|c|c|c|c|c|}\hline 
$\ttt$     & $182$& $183$   & $184$& $185$& $217$ & $218$& $219$ \\ \hline
$\frac{\text{Number of success}}{10000}$, $b=3$   & $1$ & $1$ & $9999$  & $.9997$ & $-$ & $-$& $-$\\ \hline
Time, $b=3$  & $118.710$  & $116.330$& $115.590$ & $114.55$& $-$& $-$ & $-$\\ \hline

$\frac{\text{Number of success}}{10000}$, $b=2$    & $1$ & $1$ & $1$ & $1$ & $1$& $1$& $.9997$\\ \hline
Time, $b=2$     & $339.940$& $336.990$& $332.980$& $328.040$ & $223.760$& $234.300$& $233.560$ \\ \hline
\end{tabular}
\label{RS-63-14-m1}
\end{center}
\end{table}

\begin{table}[htb]\caption{List decoding $[315,14,302]$ RS code over $\mathbb{F}_{2^9}$, multiplicity 1}
\begin{center}
\begin{tabular}{|c|c|c|c|}\hline 
$\ttt$  &$229$   & $230$   & $231$  \\ \hline
$\frac{\text{Number of success}}{10000}$     & $1$ & $1$ &$.3662$\\ \hline
Time   & $250473$ &$252184$&$258162$  \\ \hline
\end{tabular}
\label{RS-315-14}
\end{center}
\end{table}

\item Consider as block code the RS code with parameters $[63,40,24]$ over $\mathbb{F}_{2^6}$.  The repeated code with $\ell=5$ has parameters $[315,40,120]$.  Our simulations in Tables~\ref{RS-63-40} and~\ref{RS-63-40-m1} show that we can uniquely decode about $153$ errors using multiplicity assignment  (1), $110$  errors using multiplicity assignment  (2) with $b=3$ (resp $150$ errors with $b=2$), but the algorithm using multiplicity assignment (2)  is about $4.4$ times faster with $b=3$ (and $3.2$ times faster with $b=2$) than the algorithm using multiplicity assignment (1).
 
We compare the previous algorithms with  list decoding of the  RS code over $\mathbb{F}_{2^9}$ with parameters $[315,40,276]$.  With multiplicity one we can decode more  errors ($177$, see Table~\ref{RS-315-40}) than with multiplicity assignment (1) but it is $116$ times slower than the one with multiplicity assignment (1) and it is $372$ times slower than the one with multiplicity assignment (2) and $b=2$.

\begin{table}[htb]\caption{List decoding $[315,40,120]$ repeated code over $\mathbb{F}_{2^6}$  with constituent  $[63,40,24 ]$ RS code, multiplicity assignment (1)}
\begin{center}
\begin{tabular}{|c|c|c|c|c|c|c|}\hline 
$\ttt$   & $153$  & $154$ & $155$& $156$& $157$& $158$ \\ \hline
$\frac{\text{Number of success}}{10000}$  & $1$ & $.9999$   & $.9999$  & $.9998$& $.9999$& $.9999$\\ \hline
Time  & $1453.680$ & $1454.820$& $1458.830$ & $1455.610
$& $1454.390$& $1450.540$\\ \hline
\end{tabular}
\label{RS-63-40}
\end{center}
\end{table}

\begin{table}[htb]\caption{List decoding $[315,40,120]$ repeated code over $\mathbb{F}_{2^6}$  with constituent $[63,40,24 ]$ RS code, multiplicity assignment (2)}
\begin{center}
\begin{tabular}{|c|c|c|c|c|c|c|c|c|}\hline 
$\ttt$    & $110$& $111$  & $112$ & $113$& $111$& $149$&  $150$& $151$\\ \hline
$\frac{\text{Number of success}}{10000}, b=3$ & $1$  & $.9999$& $1$ & $1$& $.9997$ & $-$ & $-$ & $-$ \\ \hline
Time, $b=3$    & $327.160$& $322.680$ & $318.270$& $315.290$& $309.640$& $-$& $-$ & $-$\\ \hline

$\frac{\text{Number of success}}{10000}, b=2$     & $1$ & $1$ & $1$ & $1$& $1$& $1$& $1$& $.9999$\\ \hline
Time, $b=2$ & $559.470$   & $557.950$  & $554.180$ & $553.580$& $553.720$& $462.040$& $455.150$& $452.330$\\ \hline
\end{tabular}
\label{RS-63-40-m1}
\end{center}
\end{table}

\begin{table}[htb]\caption{List decoding $[315,40,276] $ RS code over $\mathbb{F}_{2^9}$, multiplicity 1}
\begin{center}
\begin{tabular}{|c|c|c|c|c|}\hline 
$\ttt$     & $175$   & $176$ & $177$ & $178$\\ \hline
$\frac{\text{Number of success}}{10000}$     & $1$ & $1$ & $1$ & $.2912$ \\ \hline
Time   & $161480$   & $169346$  & $169805$ & $151520$\\ \hline

\end{tabular}
\label{RS-315-40}
\end{center}
\end{table}

\item Consider as block code the RS code with parameters $[63,54,10]$ over $\mathbb{F}_{2^6}$.  The repeated code, with $\ell=5$, has parameters $[315,54,50]$.  Our simulations in Tables~\ref{RS-63-54} and~\ref{RS-63-54-m1} show that we can uniquely decode about $94$ errors using multiplicity assignment  (1) and $62$ errors using multiplicity assignment  (2) with $b=3$ (resp $89$ with $b=2$). However, the algorithm using muliplicity assignment (2) with $b=3$ is about $2.5$ times faster ($2.2$ times faster with $b=2$) than the algorithm using multiplicity assignment (1).

We compare the previous algorithms with  list decoding  of the  RS code over $\mathbb{F}_{2^9}$ with parameters $[315,54,262]$.  Using multiplicity one, we can decode more errors, about $156$, than with multiplicity assignment (1) but it is $89$ times slower than the one with multiplicity assignment (1) and it is  $207$ times slower than with multiplicity assignment method (2) and $b=2$, see  Table~\ref{RS-315-54}.

\begin{table}[htb]\caption{List decoding $[315,54,50]$ repeated code over $\mathbb{F}_{2^6}$  with constituent $[63,54,10 ]$ RS code, multiplicity assignment (1)}
\begin{center}
\begin{tabular}{|c|c|c|c|c|c|c|c|}\hline 
$\ttt$      & $93$ & $94$ & $95$& $96$& $97$& $98$& $99$\\ \hline
$\frac{\text{Number of success}}{10000}$ & $1$ & $1$   & $.9999$ & $.9999$ & $1$ & $.9999$& $.9997$\\ \hline
Time     & $1379.150$  & $1372.370$ & $1381.790$ & $1388.470$ & $1398.770$& $1405.480$& $1418.600$\\ \hline
\end{tabular}
\label{RS-63-54}
\end{center}
\end{table}

\begin{table}[htb]\caption{List decoding $[315,54,50]$ repeated code over $\mathbb{F}_{2^6}$  with constituent $[63,54,10 ]$ RS code, multiplicity assignment (2)}
\begin{center}
\begin{tabular}{|c|c|c|c|c|c|c|c|}\hline 
$\ttt$       & $60$   & $61$ & $62$& $64$& $88$ & $89$& $90$\\ \hline
$\frac{\text{Number of success}}{10000}, b=3$    & $1$ & $1$ & $1$& $.9997$&  $-$ &$-$ & $-$\\ \hline
Time, $b=3$   & $539.10$& $540.940$ & $530.780$& $527.040$& $-$ &$-$ & $-$\\ \hline

$\frac{\text{Number of success}}{10000}, b=2$  & $1$ & $1$ & $1$ & $1$ & $1$ & $1$ & $.9999$\\ \hline
Time, $b=2$ & $614.080$ & $618.220$  & $613.190$  & $611.620$ & $607.170$ & $605.460$ & $596.280$\\ \hline
\end{tabular}
\label{RS-63-54-m1}
\end{center}
\end{table}

\begin{table}[htb]\caption{List decoding $[315,54,262] $ RS code over $\mathbb{F}_{2^9}$, multiplicity 1}
\begin{center}
\begin{tabular}{|c|c|c|c|c|}\hline 
$\ttt$   & $150$ & $155$& $156$ & $157$\\ \hline 
$\frac{\text{Number of success}}{10000}$    & $1$ & $1$ & $1$& $.2657$ \\ \hline
Time      & $115934$ &$124500$ & $125797$& $129186$ \\ \hline
\end{tabular}
\label{RS-315-54}
\end{center}
\end{table}

\item  
Finally, we consider an example over a field with characteristic $3$. Consider as block code the RS code with parameters $[26,14,13]$ over $\mathbb{F}_{3^3}$.  The repeated code with $\ell=5$ has parameters $[130,14,65]$.  Our simulations in tables \ref{RS-26-14} and \ref{RS-26-14-m1} show that we can uniquely decode $65$ errors using multiplicity assignment  (1) and $46$ errors using multiplicity assignment  (2) with $b=3$ (resp $53$ with $b=2$). However, the algorithm using multiplicity assignment (2) with $b=3$ is about $4.7$ times faster ($2,7$ times faster with $b=2$) than the algorithm using multiplicity assignment (1).
 
We compare the previous algorithms with  list decoding  of the  RS code over $\mathbb{F}_{3^5}$ with parameters $[130,14,117]$.  Using multiplicity one, we can decode more errors ($77$ errors, see table \ref{RS-130-14}) than with multiplicity assignment (1) but it is $40$ times slower than the one with multiplicity assignment (1) and it is $106$ times slower than with multiplicity assignment  (2) and $b=2$, see  Table \ref{RS-130-14}.

\begin{table}[htb]\caption{List decoding $[130,14,65]$ repeated code over $\mathbb{F}_{3^3}$  with constituent  $[26,14,13 ]$ RS code, multiplicity assignment (1)}
\begin{center}
\begin{tabular}{|c|c|c|c|c|c|}\hline 
$\ttt$     & $62$ & $63$ & $64$ & $65$& $66$\\ \hline
$\frac{\text{Number of success}}{10000}$ & $1$ & $1$   & $1$ & $1$ & $.9998$ \\ \hline
Time       & $254.390$ & $254.590$ & $255.620$ & $254.980$& $256.400$\\ \hline
\end{tabular}
\label{RS-26-14}
\end{center}
\end{table}

\begin{table}[htb]\caption{List decoding $[130,14,65]$ repeated code over $\mathbb{F}_{3^3}$  with constituent $[26,14,13 ]$ RS code, multiplicity assignment (2)}
\begin{center}
\begin{tabular}{|c|c|c|c|c|c|c|c|}\hline 
$\ttt$       & $46$   & $47$ & $48$& $49$& $52$ & $53$& $54$\\ \hline
$\frac{\text{Number of success}}{10000}, b=3$    & $1$ & $.9999$ & $1$& $.9999$&  $-$ &$-$ & $-$\\ \hline
Time, $b=3$   & $55.360$& $49.280$ & $51.850$& $45.150$& $-$ &$-$ & $-$\\ \hline

$\frac{\text{Number of success}}{10000}, b=2$  & $-$ & $-$ & $-$ & $-$ & $1$ & $1$ & $.9999$\\ \hline
Time, $b=2$ & $-$ & $-$  & $-$  & $-$ & $97.230$ & $96.940$ & $96.540$\\ \hline
\end{tabular}
\label{RS-26-14-m1}
\end{center}
\end{table}

\begin{table}[htb]\caption{List decoding $[130,14,117] $ RS code over $\mathbb{F}_{3^5}$, multiplicity 1}
\begin{center}
\begin{tabular}{|c|c|c|c|c|}\hline 
$\ttt$   & $75$ & $76$& $77$ & $78$\\ \hline 
$\frac{\text{Number of success}}{10000}$    & $1$ & $1$ & $1$& $0.2769$ \\ \hline
Time      & $10055$ &$10244$ & $10320$& $10193$ \\ \hline
\end{tabular}
\label{RS-130-14}
\end{center}
\end{table}
\end{enumerate}
 
We consider a list decoding algorithm with multiplicity~1 for the RS codes due to the fact that  the times obtained in tables \ref{RS-315-14}, \ref{RS-315-40}, \ref{RS-315-54}, \ref{RS-130-14} show that higher multiplicity would be impracticable.
 These experiments clearly show advantages to using a repeated code and our method for decoding as compared to using a low rate  Reed-Solomon code.  The decoding complexity is much lower  and the correction capability is quite similar ---even better in the very low rate example---  for the repeated code. Although the repeated code  does not have   good parameters, it is  because a rare few codewords are ``close" to the sent codeword.  However,  these have little affect on the decoding performance.

\section{Bounds for the correction capability using multiciplity assignment (2)}\label{sec:boundsfortau2}

 We consider the decoding capability of this algorithm  for  the repetition code of a Reed-Solomon code with   $b=  \lfloor \ell /2 \rfloor +1$ and $b=  \lfloor \ell /2 \rfloor$. The decoding capability of the soft-decoding algorithm for Reed-Solomon codes can be found in \cite[Section IV]{Lee-Michael1}. However, the bounds are obtained in terms of the interpolation multiplicities and we cannot infer a bound in terms of the weight of the error vector. In order to obtain such a bound, we should assume that the error vector verifies an additional condition for performing our analysis.

\begin{theorem}\label{te:BoundsAssignment2}
The algorithm  introduced in Section~\ref{sec:interp-repe}  using the second multiplicity assignment for a Reed-Solomon repetion code $[n\ell,k,\ell (n-k+1)]$   can decode at least the following number of errors if the error vector $\mathbf{e}$ verifies the following assumption: for every $i \in \{1, \ldots, n\}$ and $ \beta \in \mathbb{F}_q\setminus  \{ 0\}$,
 $$\# \{ e_i^j = \beta : j = 1 , \ldots , \ell\} \le b-1.$$ 
\begin{itemize}
\item For $b=  \lfloor \ell /2 \rfloor +1$ and $\ell$ odd,
\[
(n-k)\left( \left\lfloor \frac{\ell}{2} \right\rfloor +1 \right)+ \left\lfloor \frac{\ell}{2} \right\rfloor~\mathrm{errors}.
\]
\item For $b=  \lfloor \ell /2 \rfloor +1$ and $\ell$ even,
\[
(n-k)\left\lfloor \frac{\ell}{2} \right\rfloor  +\left\lfloor \frac{\ell}{2} \right\rfloor-1~\mathrm{errors} .
\]
\item For $b=  \lfloor \ell /2 \rfloor $ and $\ell$  odd,
\[
(n-k)\left(\left\lfloor \frac{\ell}{2} \right\rfloor +2 \right)+\left(\left\lfloor \frac{\ell}{2} \right\rfloor +1 \right)~\mathrm{errors}.
\]
\item For $b=  \lfloor \ell /2 \rfloor $ and $\ell$ even,
\[
(n-k)\left(\left\lfloor \frac{\ell}{2} \right\rfloor +1 \right)  +\left\lfloor \frac{\ell}{2} \right\rfloor~\mathrm{errors}.
\]
\end{itemize}
\end{theorem}
\begin{proof}
 Notice that the constituent code  is a MDS code. Under the assumption that   $\# \{ e_i^j = \beta : j = 1 , \ldots , \ell\} \le b-1$ for every $i \in \{1, \ldots, n\}$ and $ \beta \in \mathbb{F}_q\setminus  \{ 0\}$ we guarantee that we only have erasures (we have no errors) and we assign to them  multiplicity zero. 
An erasure MDS code can correctly decode a received word if $k$ symbols are non-corrupted, in our setting, this means that we have assigned multiplicity one.  The point $(\alpha^i,r_i^j)$ is assigned multiplicity one if and only if $b$  blocks have the same value in the $i$-th position. Therefore  $(n-k)$ positions can be corrupted in $\lfloor \ell/2 \rfloor +1$ blocks and we still  can decode it correctly if $b=  \lfloor \ell /2 \rfloor +1$ and $\ell$ is odd, $\lfloor \ell/2 \rfloor $ if $b=  \lfloor \ell /2 \rfloor +1$ and $\ell$ is even, 
$(\lfloor \ell/2 \rfloor +2 )$ if $b=  \lfloor \ell /2 \rfloor $ and $\ell$ is odd, $(\lfloor  \ell / 2  \rfloor +1 )$ if $b=  \lfloor \ell /2 \rfloor $ and $\ell$ is even. 

In the other $k$ positions there could be at most $\lfloor \ell/2 \rfloor$ errors to guarantee success, i.e. to obtain an erasure in such position, if $b=  \lfloor \ell /2 \rfloor +1$ and  $\ell$ is odd, $\lfloor \ell/2 \rfloor -1 $ if $b= \lfloor \ell /2 \rfloor +1$ and $\ell$ is even,  $\lfloor \ell/2 \rfloor+1$ if $b= \lfloor \ell /2 \rfloor$ and $\ell$ is odd, $\lfloor \ell/2 \rfloor$ if $b= \lfloor \ell /2 \rfloor$ and $\ell$ is even. 

Summing up all these values we get the above bound.
\end{proof}

\begin{remark}\label{re:examples}

Let us compute these bounds for the examples in previous section.
\begin{itemize}

\item For the $[63,14,50]$ RS code and $\ell=5$ and $b=3$ Theorem \ref{te:BoundsAssignment2} tell us that we can decode  $149$ errors  while if $b=2$ we can decode $199$ errors. If we check the computer experiments in Table \ref{RS-63-14-m1} we can successfully decode $183$ errors for $b=3$ and $218$ errors for $b=2$ respectively.
\item For the $[63,40,50]$ RS code and $\ell=5$ and $b=3$ Theorem \ref{te:BoundsAssignment2} tell us that we can decode  $71$ errors  while if $b=2$ we can decode $95$ errors. If we check the computer experiments in Table \ref{RS-63-40-m1} we can successfully decode $113$ errors for $b=3$ and $150$ errors for $b=2$ respectively.

\item For the $[63,54,50]$ RS code and $\ell=5$ and $b=3$ Theorem \ref{te:BoundsAssignment2} tell us that we can decode  $29$ errors  while if $b=2$ we can decode $39$ errors. If we check the computer experiments in Table \ref{RS-63-54-m1} we can successfully decode $61$ errors for $b=3$ and $89$ errors for $b=2$ respectively.
\item 
For the $[26,14,13]$ RS code and $\ell=5$ and $b=3$ Theorem \ref{te:BoundsAssignment2} tell us that we can decode  $38$ errors  while if $b=2$ we can decode $51$ errors. If we check the computer experiments in Table \ref{RS-26-14-m1} we can successfully decode $53$ errors for $b=2$ and $46$ errors for $b=3$ respectively.
\end{itemize}
Therefore, we claim that we can use our bounds as a conservative estimate of the real decoding capacity, especially for low rate codes, independently of the assumptions.
\end{remark}

Note that our assumption  in Theorem 1  for the error vector $\mathbf{e}$ will hold with a high probability if the field is not too small:  the higher the base field is, the more unlikely is that two error positions $e_i^j$ and $e_{i}^{j'}$ are equal, for $i=1,\ldots , n$. We compute now when this algorithm can decode more errors than a $[n\ell,k,n\ell-k+1]$ Reed-Solomon code with a unique decoding algorithm. Notice that we perform list decoding of a Reed-Solomon code over $\mathbb{F}_q$ and we compare it with a unique decoding of  a $[n\ell,k,n\ell-k+1]$ MDS  code over $\mathbb{F}_{q}$ (if it exists) or over a higher field.

\begin{itemize}
\item If $b=  \lfloor \ell /2 \rfloor +1$ and $\ell$ is odd, i.e., $\ell=2\ell'+1$ where $\ell' = \lfloor \ell /2  \rfloor$.

\[
 (n-k)\left(\left\lfloor \frac{\ell}{2} \right\rfloor +1 \right)+\left\lfloor \frac{\ell}{2} \right\rfloor \geq \left\lfloor \frac{n\ell-k+1}{2} \right \rfloor
\geq \frac{n\ell-k+1}{2}
\]
We compare the  leftmost and rightmost sides of the formula. We have that
\[
2(n-k)(\ell'+1)+2\ell'\geq  2n\ell'+n-k+1
\]

\[
2n \ell' + 2n - 2k(\ell' +1) + 2\ell' \ge 2n \ell' + n-k+1 
\]

\[
n + 2\ell'-1 \ge k (2\ell' +1) 
\]
\[
k \le \frac{n + 2\ell'-1}{\ell}
\]

For example if $n=63$, $\ell=5$, $\ell'=2$ then for $k\leq 13$ 
we can correctly decode at least up to the half of the minimum distance  of the corresponding Reed-Solomon code. According to our computations for $k=14$ we can indeed decode about $187$ errors but a Reed-Solomon code can correct 150 errors.
 \end{itemize}
 
Analogously, we have that: 

\begin{itemize}
\item If $b=  \lfloor \ell /2 \rfloor +1$ and $\ell$ is even, 
 i.e., $\ell=2\ell'$ where $\ell' = \lfloor  \ell/2  \rfloor$, then 

\[
k \le \frac{ 2\ell'-3}{2\ell' -1}<1.
\]
Therefore, we can conclude that one should not consider $\ell$ even for $b \ge  \lfloor \ell /2 \rfloor +1$.

\item If $b=  \lfloor \ell /2 \rfloor $ and $\ell$ is odd, i.e., 
 $\ell=2\ell'+1$ where $\ell' = \lfloor \ell/2  \rfloor$, then 

\[
k \le \frac{3n + 2\ell'+1}{2\ell' +3}
\]

For example if $n=63$, $\ell=5$, $\ell'=2$ then for $k\leq 27$ 
we can correctly decode at least up to the half of the minimum distance  of the corresponding Reed-Solomon code. According to our computations for $k=14$ we can indeed decode about $219$ errors but a Reed-Solomon code can correct 150 errors.

\item If $b=  \lfloor \ell /2 \rfloor $ and $\ell$ is even, i.e., 
 $\ell=2\ell'$ where $\ell' = \lfloor \ell/2  \rfloor$, then

\[
k \le \frac{ 2n+ 2\ell' -1}{2\ell'+1}.
\]
For example if $n=63$, $\ell=4$, $\ell'=2$ then for $k\leq 25$ 
we can correctly decode at least up to the half of the minimum distance  of the corresponding Reed-Solomon code. 

\end{itemize}

\section{Bounds for the correction capability using multiciplity assignment (1)}
\label{sec:boundsfortau}

As mentioned in Section \ref{sec:interpolating},  the interpolation
problem consists in finding a bivariate polynomial, $Q(x,y)$, passing
through the points $p_{i,\beta}$ with multiplicity $m_{i,\beta}$,
where the multiplicities are described in section
\ref{sec:interp-repe}. Therefore, we compute a Gr\"obner basis of the
ideal $I_M$, defined in (\ref{def: I_M}), with respect to the
$(1,k-1)$-weighted degree and consider the smallest element in this
basis. Succesful decoding is ensured when for some integer $d$ the following two conditions are satisfied:  

\begin{itemize}
\item[(i)]  The number of monomials of $(1,k-1)$ degree at most  $d$
 is larger than  the number of conditions imposed.

\item[(ii)]  The score of the sent codeword $c$ is larger than $d$.
\end{itemize}
The first item ensures the existence of a polynomial $f \in I_M$ of
weighted degree at most $d$ and the second ensures that $y-h_c$ is a
factor of $f$.

In the rest of this section we analyze when both conditions are
simultaneously satisfied in the context of quasi-cyclic codes with
multiplicity assignment (2).   Let $\mathbf{r}$ be the received word,
let $ \ttt$ be the number of errors, and for each $i=1,\dots,n$ let
$\ttau_i$ be the number of errors in position $i$.  Thus, $\sum_{i=1}^n 
\ttau_i=\ttt$.

Let us start with condition (ii).    By Lemma~\ref{le:scorebound}  the
score of $\mathbf{r}$ is $\ell n -\ttt$.  Let $a,b$  be the unique
integers satisfying $0\leq b < k-1$ and 
\begin{align*}
\ell n-\ttt -1 = a(k-1)+b
\end{align*}

Now consider condition (i). 
For $a\geq 0$ and $0\leq b \leq k-2$ let $P_{a,b}$ be the set of
monomials with $(1,k-1)$-weighted degree lower than or equal to
$a(k-1)+b$.  
The number of monomials in $P_{a,b}$ is
\begin{equation}\label{eq:NumberOfMonomials}
(k-1)+2(k-1)+\cdots+a(k-1)+(a+1)(b+1)=\frac{a(a+1)}{2}(k-1)+(a+1)(b+1).
\end{equation}
Thus, comparing   (\ref{eq:NumberOfMonomials}) with (\ref{eq:nequat})  we get the following condition:

\begin{equation} \label{eq:exist}
(k-1)\frac{a(a+1)}{2}+(b+1)(a+1)\geq \sum_{i=1}^{n}\sum_{\beta\in\mathbb{F}_q}\binom{m_{i,\beta}+1}{2}+1
\end{equation}
where $M$ is the multiplicity matrix derived from $\mathbf{r}$
according to multiplicity assignment (2).
In order to analyze the multiplicity matrix we will make a simplifying assumption:
If two errors occur at position $i$ and blocks  $j,k$ then
$r_{i}^j\neq r_{i}^k$.  That is, 
\begin{equation}\label{as:e1}e_{i}^j \neq e_{i}^k.\end{equation}
 If the base field is large and  $\ell$ is small, this assumption is realistic.

Under the above assumption, for position $i$ there are $\ell-\tau_i$
blocks that have the correct value and $\ttau_i$  that are incorrect,
but unequal to one another.  The  number of conditions imposed by the correct blocks is 
is $\binom{\ell-\ttau_i}{2}$, while the incorrect blocks impose one
condition each, for a total of $\ttau_i $.    Thus the total number of
conditions imposed is 
\begin{align*}
\sum_{i=1}^{n} \Big( \binom{\ell-\ttau_i}{2} + \ttau_i \Big) &= 
\sum_{i=1}^{n}\Big( \dfrac{(\ell-\ttau_i)(\ell -\ttau_i+1)}{2} + \ttau_i \Big)\\
&= \sum_{i=1}^{n} \Big(\dfrac{(\ell)(\ell+1)}{2}  - \ttau_i(\ell-1) + \dfrac{\ttau_i(\ttau_i-1)}{2}\Big) \\
&= n\dfrac{(\ell)(\ell+1)}{2}  - \ttt(\ell-1) + \sum_{i=1}^{n}\dfrac{\ttau_i(\ttau_i-1)}{2} 
\end{align*}

Notice that the final sum $A=
\sum_{i=1}^{n}\dfrac{\ttau_i(\ttau_i-1)}{2} $ is the only one that depends on the
distribution of the errors.  We now consider  three cases.  
The  term $A$ is minimized when the
$\tau_i$ are distributed as evenly possible.  Let $\ttt = nq_1 + s_1$
with $0 \leq s_1<n$ and
assume that $s_1$ positions have $q_1+1$ errors while $n-s_1$ positions have 
$q_1$ errors.  The final term is then
\[
A_{\min} = s_1\dfrac{(q_1+1)q_1}{2} + (n-s_1) \dfrac{(q_1-1)q_1}{2} = n \dfrac{q_1(q_1-1)}{2} +s_1q_1 
\]
The final term $A$ is maximized when the  errors are
consolidated into as few  positions as possible.
Let $\ttt = \ell q_2 + s_2$ with $0\leq s_2 < \ell$ and assume that $q_2$ positions have
$\ell$ errors and that one position has $s_2$ errors.  The final term
is then 
\[
A_{\max} = q_2 \dfrac{\ell (\ell -1)}{2} + \dfrac{s_2(s_2-1)}{2}
\]

Finally we consider the expected value of the final term $A$, subject to
the $\ttt$ error positions being randomly chosen from 
$\{(i,j): i \in \{1,\dots,n\}, j \in \{1,\dots, \ell\} \}$.
The probability of any particular $\ttau_1, \dots, \ttau_n$ occuring
is 
\[ \dfrac{ \prod_{i=1}^n \binom{\ell}{\ttau_i}}{ \binom{n\ell}{\ttt}}.\]
Let $\mathbf{x} $ stand for indeterminates $x_1, \dots, x_n$
and let 
$\mathbf{1}$ be an $n$-tuple with 1 in each entry.  
Let $|\ttau|= \sum_{i=1}^n \ttau_i$.
Consider the generating function 
\begin{align}
\label{e:B}
B(\mathbf{x},s) = \prod_{i=1}^n(1+sx_i)^\ell
\end{align}
The term in  $s^\ttt$ in $B(\mathbf{x},s)$  is 
\[\sum_{|\ttau| = t} \prod_{i=1}^n \binom{\ell}{\ttau_i}
(sx_i)^{\ttau_i}
= s^\ttt\sum_{|\ttau| = t} \prod_{i=1}^n \binom{\ell}{\ttau_i} (x_i)^{\ttau_i}
\]
 Taking the second derivatives of $B(\mathbf{x},s)$ with respect to the
 $x_i$ one finds that  the term in  $s^\ttt $  in
 $\sum_{i=1}^n \dfrac{\partial^2 B}{\partial x_i^2} (\mathbf{x},s)$ is 
 \[
 s^\ttt  \sum_{|\ttau| = t} \Big( \prod_{i=1}^n \binom{\ell}{\ttau_i} \Big) \sum_{i=1}^n
 \tau_i(\tau_i-1)(x_i)^{\ttau_i-2}
 \]
 Evaluating at $\mathbf{x}= \mathbf{1}$ we get the term in $s^t$ in 
 $\sum_{i=1}^n \dfrac{\partial^2 B}{\partial x_i^2} (\mathbf{1},s)$ is 
 \[
 s^\ttt  \sum_{|\ttau| = t} \Big( \prod_{i=1}^n \binom{\ell}{\ttau_i} \Big) \sum_{i=1}^n
 \tau_i(\tau_i-1)
 \]
 Thus the expected value of the final term $A$ is the coefficient of
$s^\ttt$ in 
$(2\binom{n\ell}{\ttt})^{-1}\sum_{i=1}^n \dfrac{\partial^2   B}{\partial x_i ^2} (\mathbf{1},s)$

Computing from \eqref{e:B} we have
\begin{align*}
\sum_{i=1}^n \dfrac{\partial^2 B}{\partial x_i^2} (\mathbf{x},s) 
&=  s^2\sum_{i=1}^n \ell(\ell-1)(1+sx_i)^{\ell-2} 
\prod_{\substack{  j=1 \\ j \ne i}} ^n (1+sx_j)^\ell \\
\sum_{i=1}^n \dfrac{\partial^2 B}{\partial x_i^2}(\mathbf{1},s) 
&=  s^2\sum_{i=1}^n \ell(\ell-1)(1+s)^{n\ell-2} \\
&=  s^2n \ell(\ell-1)(1+s)^{n\ell-2} 
\end{align*}
The coefficient of $s^\ttt$ in  
$(2\binom{n\ell}{\ttt})^{-1}\sum_{i=1}^n \dfrac{\partial^2
  B}{\partial x_i^2} (\mathbf{1},s)$ is therefore
\begin{align*}
 A_{\exp} =  n\ell(\ell-1) \dfrac{(n\ell-2)!}{(\ttt-2)!(n\ell-\ttt)!}  
\dfrac{ (n\ell-\ttt)!\ttt!} {2     (n\ell)!}
&=
\dfrac{\ttt(\ttt-1)(\ell-1)}{2(n\ell-1)}
\end{align*}

In the following table for each code considered in Section~\ref{sec:sim} we show the
maximum value $t$ for which the algorithm is  able to decode in
three cases: for the worst scenario, for the expected
scenario, and for the best scenario (i.e. using $A_{\max}$,
$A_{\exp}$, $A_{\min}$).
We also list 
the range where experiments showed that decoding capability began to decline.
It is interesting to see that the experimental results are somewhat
better, even better than the best scenario (using $A_{\min}$), particularly at higher
rate.   
In the simulations, the number of monomials less than the leading term of the minimal  $Q(x,y)$ that was computed is noticeably smaller than the number of conditions imposed.  In other words, the conditions imposed by interpolation are not independent.
We have no explanation for this phenomenon, but it appears to be  the key to the performance beyond our estimates.

\begin{table}[htb]\caption{Bounds for the correction capability using multiplicity assignment (1)}
\begin{center}
\begin{tabular}{|c|c|c|c|c|}\hline 
code  & max & exp &min  & range\\ \hline 
$[63,14]_{\mathbb{F}_{2^6}}$ & 203&223&227& 226-230 \\\hline
$[63,40]_{\mathbb{F}_{2^6}}$ &89 &130 & 142& 153-158\\\hline
$[63,54]_{\mathbb{F}_{2^6}}$ &38 &69 &81 &94-99 \\ \hline
$[26,14]_{\mathbb{F}_{3^3}}$ &49 &64 &69 &65-66 \\ \hline
\end{tabular}
\end{center}
\end{table}

\section{Conclusion}\label{sec:con}

An efficient list-decoding algorithm for repeated codes, in particular
for repeated Reed-Solomon codes is presented. The theoretical and experimental
results show that decoding repeated codes with this algorithm yields
surprisingly good error correction performance, nearly comparable to
that of a Reed-Solomon codes over a larger field.
Furthermore, the  computational burden of the repeated code is much
lower because of the smaller field size.

\end{document}